\documentclass[11pt,a4paper]{article}
\usepackage[T1]{fontenc}
\usepackage[utf8]{inputenc}
\usepackage{authblk}
\usepackage{hyperref}
\usepackage{cite}
\usepackage{bbm}
\usepackage{graphicx}
\usepackage{amsfonts}
\usepackage{amssymb}
\usepackage{amsmath}
\usepackage{amsthm}
\usepackage{subfig}
\usepackage{float}
\usepackage{color}
\usepackage{mdwlist}
\usepackage[normalem]{ulem}
\usepackage{listings}
\DeclareUnicodeCharacter{00A0}{ }

\newcommand{\R}{\mathbb{R}}
\newcommand{\C}{\mathbb{C}}
\renewcommand{\H}{\mathbb{H}}
\newcommand{\1}{\mathbbm{1}}

\renewcommand{\P}{\mathbb{P}}

\newcommand{\M}{\mathcal{M}}
\newcommand{\I}{\mathcal{I}}

\newcommand{\avg}[1]{\langle #1 \rangle}

\definecolor{orange}{rgb}{1.00,0.50,0.0}
\newcommand{\comment}[1]{{\color{black}{#1}}}

\theoremstyle{plain}
\newtheorem{thm}{Theorem}
\newtheorem{proposition}{Proposition}
\newtheorem{corollary}{Corollary}

\theoremstyle{remark}
\newtheorem{remark}{Remark}
\theoremstyle{definition}

\graphicspath{{../figures/}}

\begin{document}

\title{The real Ginibre ensemble with $k=O(n)$ real eigenvalues}

\author[1,2]{Luis Carlos Garc\'ia del Molino}
\author[1]{Khashayar Pakdaman}
\author[2]{Jonathan Touboul}
\author[3]{Gilles Wainrib}

\affil[1]{Institut Jacques Monod, CNRS UMR 7592, Universit\'e Paris Diderot, Paris Cit\'e Sorbonne, F-750205, Paris, France}
\affil[2]{Mathematical neuroscience Team, CIRB-Coll\`ege de France\footnote{CNRS UMR 7241,
INSERM U1050, UPMC ED 158, MEMOLIFE PSL*} and INRIA Paris-Rocquencourt, MYCENAE Team, 11 place Marcelin Berthelot, 75005 Paris, France}
\affil[3]{Ecole Normale Sup\'erieure, D\'epartement d'Informatique (DATA), 45 rue d'Ulm, 75005 Paris, France}
\renewcommand\Authands{ and }

\date{\today}%

\maketitle

\begin{abstract}
We consider the ensemble of real Ginibre matrices conditioned to have positive fraction $\alpha>0$ of real eigenvalues. We demonstrate a large deviations principle for the joint eigenvalue density of such matrices and introduce a two phase log-gas whose stationary distribution coincides with the spectral measure of the ensemble. Using these tools we provide an asymptotic expansion for the probability $p^n_{\alpha n}$ that an $n\times n$ Ginibre matrix has $k=\alpha n$ real eigenvalues and we characterize the spectral measures of these matrices.
\end{abstract}


\section{Introduction}
Random matrices constitute a central topic in modern probability
theory~\cite{Forrester10,TaoBook12,Bordenave12}
and an important tool for an increasing number of applications, from
physics~\cite{wigner1955lower,auffinger2013random} to biology
\cite{may1972will,wainrib2013topological,Garcia13} or
engineering~\cite{couillet,jaeger,compressedsensing}.

The present study deals with specific properties of a
canonical family of random matrices, the real Ginibre
ensemble. \comment{Under a reference probability $\P$,} the entries of such matrices are i.i.d. normal random variables \comment{with variance $\frac{1}{n}$ where $n$ is the matrix size}. In particular we intend to look deeper into the properties of real Ginibre matrices with anomalously large number of
real eigenvalues, which are still largely unknown. These constraints are
drastic for the random matrices, and affect the shape of the
distribution of the eigenvalues as illustrated in Fig. \ref{fig:spec}.
We aim at characterizing the eigenvalue distribution of random
matrices when the number of real eigenvalues $k$ is proportional to the matrix size
$n$.

Our work reveals that in this regime, the empirical spectral measure (ESD)
markedly departs from the one of the unconditioned ensemble. We establish that when $n \to
\infty$, the ESD converges to a limit that is supported
by both the real line and the complex plane.  We characterize the
macroscopic properties of this limit and analyze its microscopic
organization. In the process, we also obtain an estimate for the probability that an $n\times n$ matrix has $k$ real eigenvalues $p^n_k$
with $k=O(n)$ as $n\to \infty$.

\begin{figure}[t]
\centering
 \includegraphics[width=0.6\textwidth]{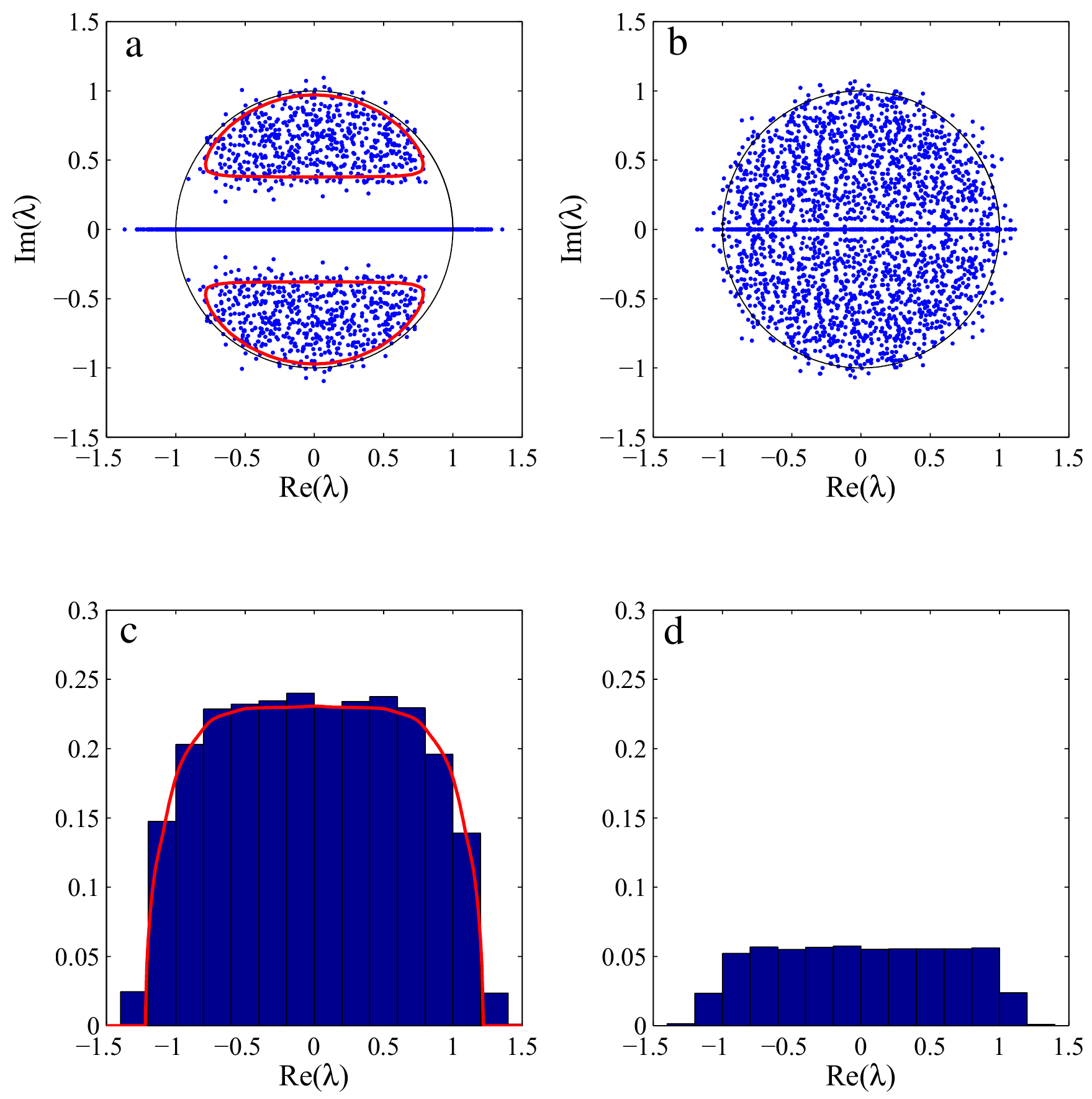}
 \caption{Top: Superposition of $50$ spectra of $50\times 50$ Ginibre matrices conditioned to have $26$ real eigenvalues (a) and unconditioned (b). Black lines correspond to the unit circle and red lines correspond to our estimation of the support of the complex part of the spectrum in the large $n$ limit. Bottom: Histogram of the real eigenvalues of $200$ $50\times 50$ Ginibre matrices conditioned to have $26$ real eigenvalues (c) and unconditioned (d). Histograms are normalized with respect to $n$. Red line corresponds to our estimation of the distribution of real eigenvalues in the large $n$ limit.}\label{fig:spec}
\end{figure}

\paragraph{Literature review.}
Before presenting our main results, we concisely review some relevant past results in random matrix theory. Characterizing the spectral properties of the real Ginibre ensemble
has been an active field of research. A milestone in this
direction was the computation, \comment{first in
\cite{Lehmann91} and later in \cite{Edelman97}}, of the joint probability distribution of
the eigenvalues \comment{$\lambda_1,\dots,\lambda_n$} of $n\times n$ real Ginibre matrices:
\begin{equation}\label{eq:pdf}
p[\lambda_1,\cdots,\lambda_n]=C_n\prod_{i>j}|\lambda_i-\lambda_j|\left(\prod_{i=1}^n\exp(-n\lambda_i^2)\textrm{erfc}\left(|\lambda_i-\lambda_i^*|\sqrt{n/2}\right)\right)^{1/2}
\end{equation}
where the asterisk denotes the complex conjugate \comment{and $C_n$ is the normalization constant}. \comment{We denote in the sequel by $Q^n$ the associated probability measure (see equation~\eqref{eq:Qn}).}

A later breakthrough
was the characterization of the correlations between
eigenvalues in terms of Pfaffian processes in a series of studies
\cite{Forrester07, Sommers07, Forrester08, ForresterMays09, Sinclair09, SommersWieczorek08, BorodinSinclair09}. More recently, this picture has been further augmented
by the description of the distribution of the spectral radius and of
the largest real eigenvalue of this ensemble of matrices
\cite{rider2012extremal}.

Much work has also been devoted to the characterization of the spectrum of real Ginibre
matrices in the limit of $n \rightarrow \infty$. We denote by $\hat \mu^n$
the empirical spectral distribution of such matrices of size $n\times
n$ defined as:
\begin{equation}
\hat\mu^n=\frac{1}{n}\sum_{i=1}^{n}\delta(\lambda_i)
\end{equation}
where $\{\lambda_i\}_{i=1}^n$ are the eigenvalues of the matrix.
For $M^n$ a real Ginibre matrix scaled by $1/\sqrt{n}$, it is now well
known that $\hat\mu^n$ converges to the uniform distribution on the
unit disk as $n\rightarrow \infty$, a result known as circular law:
\begin{equation}
\hat\mu^n\rightharpoonup \frac{1}{\pi}\1_{|x|\leq1}\ .
\end{equation}
For real Ginibre matrices, this result was first demonstrated in
\cite{Edelman97}. It has been now established that the circular law is
universal in the sense that convergence of $\hat\mu^n$ to the uniform
distribution on the unit disk holds for matrices composed of i.i.d.
random variables with zero mean and $1/n$ variance \cite{Tao10}.
Furthermore, local properties, such as correlations between
eigenvalues are also universal in the $n\rightarrow \infty$ limit \comment{for matrices with independent elements with exponentially decaying distribution and moments matching the normal distribution up to fourth order~\cite{Tao12}}.

Large Deviations Principles
(LDP) for random matrices
were derived in \cite{Benarous97,Benarous98}, where it was shown that
the sequence of empirical measures $\{\hat\mu^n\}_{n\to\infty}$ of
$M^n$ satisfy a LDP with speed $n^{2}$ and rate function
\begin{equation}\label{eq:Ialpha}
\mathcal{I}[\mu]:
\begin{cases}
	\M&\mapsto \R\\
	\mu &\mapsto \frac12\left(\int x^ 2d\mu(x) - \int\int
\log|x-y|d\mu(x)d\mu(y)\right)-K
\end{cases}
\end{equation}
for symmetric Gaussian ensembles (i.e. Gaussian Orthogonal Ensemble) and the real Ginibre ensemble
respectively. In the former, the map $\mathcal{I}$ acts on
probability measures on the real axis $\M_1^+(\R)$ and
$K=\frac 3 8+\frac 1 4 \log 2$. In the later, the map $\mathcal{I}$
acts on probability measures on $\C$ symmetrical with respect to
complex conjugation: $ \M_1^S(\C)$, and $K= \frac 3 8$.

The distinctive feature of the spectrum of real Ginibre matrices is
that it has a non-zero probability of having real eigenvalues.
As shown in \cite{Edelman94,Edelman97,Ginibre65}, the empirical
spectral distribution of finite real Ginibre matrices has a
singularity on the real line because there is a positive probability
of having real eigenvalues. As the matrix size $n$ goes to infinity,
this singularity disappears because the expected number of real
eigenvalues is of order  ${\sqrt{n}}$. The first numerical reports on
this scaling appeared in \cite{Sommers88} and a rigorous proof for the
average number of real eigenvalues and higher order expansions in
\cite{Edelman94}. Its universality was established in \cite{Tao12}. Recent studies have provided a more detailed analysis of the distribution of real eigenvalues of real Ginibre matrices, notably the inter eigenvalue gap distribution \cite{Tribe11,Forrester13,Beenakker14}, which takes on an approximately semi-Poisson form in the bulk.

In \cite{Edelman97} the author
introduced the probabilities $p^n_k$ for a real Ginibre matrix of
size $n$ to have $k$ real eigenvalues and provided numerical estimates
for some special cases of $p^n_k$ and exact expressions for the
specific case $k=n$. \comment{Integral
expressions for these probabilities for small $n$ were derived\cite{Kanzieper}, enlarging the range of values
that could be numerically computed. An exact expression for $k_n=n-2$ and its large $n$ asymptotic behavior was derived in~\cite{akemann07} using the integrable structure of the real Ginibre ensemble and Pfaffian properties. Recently, fine asymptotic estimates of $p^n_k$ for $k$ small ($k=o(\log(n)/\sqrt{n})$) were analytically derived in~\cite{kanzieper2015}. To our knowledge, there is no asymptotic expression of $p^n_k$ for large $n$ and general $k$. }

Another very efficient method for the study of the spectra of the Gaussian $\beta$ ensembles was proposed in~\cite{Dyson62}. This pioneering work made a deep analogy between 1d log-gases, i.e. freely moving charged particles
with quadratic confinement and logarithmic repulsion, and the
eigenvalues of matrix-valued symmetric real Ornstein-Uhlenbeck
processes. In detail, the equilibrium distribution of a
one-dimensional log-gas at an inverse temperature $\beta$ is
precisely the distribution of the eigenvalues of the Gaussian $\beta$
ensemble (symmetric, hermitian or quaternionic random matrices), and
the equilibrium density of a two-dimensional log-gas is identical
to the distribution of eigenvalues of the complex Ginibre ensemble.
This link between interacting particle systems and spectra of random
matrices has proved an essential tool to demonstrate properties of the
spectrum of random matrices even in cases with extremely low
probability. As an example, the use of the log-gas for symmetric
matrices was instrumental in the characterization the spectrum of
random matrices with anomalous densities~\cite{majumdar2009index,majumdar2012number}, with applications to data analysis.

From the mathematical viewpoint, the existence and uniqueness of
solutions to 1d log-gas systems as well as the convergence as the
system size goes to infinity were proved in~\cite{Rogers93} and for a
more general class of gases in~\cite{Cepa97}. For more on log-gases we refer to \cite{Forrester10}.

\paragraph{Methods and summary of the main results.}

Our methods rely on the derivation of a LDP for the Ginibre ensemble conditioned on the proportion of real
eigenvalues $\alpha=k/n$.  Our contribution
here is to extend the LDP in \cite{Benarous97,Benarous98} to the situation that interpolates between the two cases presented
 to allow for measures that are supported both on the real axis and the complex plane. We find that such matrices satisfy a LDP with rate $n^2$ and rate function $\I$, where $K=\frac{3}{8}$.

From our LDP, we are able to show that when $k/n\rightarrow \alpha$ and $n\to\infty$, $\frac{1}{n^2}\log p^n_k$ scales as $\I[\mu_\alpha]-K$ where $\mu_\alpha$ is the minimizer of the rate function
$\mathcal{I}$ on the set $\mathcal{M}_{\alpha} = \{\nu \in \M_1^S(\C)\
;\ \nu(\R)=\alpha\}$. In particular, for the case $\alpha=1$ one can see that \[\frac{1}{n^2}\log p^n_n \to-\I[\mu_1]= -\frac{1}{4}\log 2\ ,\]
which coincides asymptotically with the \comment{exact formula derived in \cite{Edelman97} and with the formula of $p_{n-2}^n$ derived in~\cite{akemann07}; actually, our result shows that this logarithmic equivalent is valid for $p_{n-2r}^n$ for any $r\in \mathbbm{N}$ and not only for $r=0$ or $1$. } However, obtaining a closed form expression for the minimizer $\mu_\alpha$ is not straightforward. Nonetheless, we are able to derive a precise qualitative picture of the support
and shape of the minimizer through the use of a constrained
optimization problem~\cite{Anderson09}. 

To gain a deeper understanding on the minimizer, we next introduce and investigate the log-gas whose stationary distribution corresponds to the eigenvalue distribution of the class of matrices we are studying. 
 In contrast with the
existing literature, this log-gas is neither one nor two dimensional:
it is a mixture, in the complex plane, of the two types of gas, one
fraction of the particles being confined on a singular region of the
plane. We use this gas to obtain numerically for various values of $\alpha$, approximations of the distribution $\mu_\alpha$. 

The above steps characterize the macroscopic properties of the limit distribution $\mu_\alpha$. We complement these by a description of microscopic features. To this end, we use renormalization techniques inspired from hydrodynamics and the Ginzburg-Landau theory \cite{sandier20131d,Sandier13,rougerie2013higher} that were extended to investigate the microscopic organization of particles in log-gases with applications to the distribution of eigenvalues of the Ginibre ensemble. This approach has unveiled in particular a crystallization phenomenon in one dimension, and led to conjecture that particles in two-dimensional log-gases organize according to a regular triangular lattice in the zero temperature limit~\cite{Sandier13,sandier20131d}. We readily apply these methods to our mixed-type problem.

\paragraph{Organization of the paper.}
Since the complex eigenvalues of real matrices come in pairs of complex conjugates, matrices where $k$ and $n$ have different parity have probability 0. For this reason through the text we assume that the number of real eigenvalues $k$ has the same parity as $n$.  Also we introduce the notation $\overline{\alpha n}$ for the closest integer to $\alpha n$ with the same parity as $n$. 

The article is organized as follows. In section \ref{sec:LDP}, we
establish a specific LDP for real Ginibre matrices conditioned to
have $k=\overline{\alpha n}$ real eigenvalues and the asymptotic estimation of
$p^n_k$ in section \ref{sec:asymptotic_pnk}. We characterize more precisely the form of the distribution of real and complex eigenvalues minimizer in section \ref{sec:variational}. We introduce and analyze in
section \ref{sec:particlesystem} the 1d 2d log-gas whose stationary distribution is identical to the eigenvalues of a Ginibre matrix constrained on having a specific number of real eigenvalues. Finally, in section \ref{sec:renorm}, we derive the renormalized energy for the mixed gas and discuss its implications in terms of the distributions of the particles in the zero temperature limit.

\section{Large deviations principle for $k=\overline{\alpha n}$ and analysis of the rate function}\label{sec:LDP}

Consider $M^n_{k_n} \in \R^{n\times n}$ with $n\in\mathbb{N}$ a sequence of real Ginibre random matrices with $k_n$ real eigenvalues, such that $k_n$ has the same parity as $n$ and $k_n/n \to \alpha \in [0,1]$. 
The large deviations principles shown in \cite{Benarous97,Benarous98} correspond to $\alpha\in\{0,1\}$. We now show that they can be extended to $\alpha\in (0,1)$. To this purpose, we define $\M_{\alpha}$ as the subset of symmetrical probability measures $\M_1^S(\C)$ exactly charging a mass $\alpha$ to the real line:
\[\mathcal{M}_\alpha:=\{\mu \in \M_1^S(\C):\mu(\R)=\alpha\}\]
For fixed $\alpha\in (0,1)$ and finite $n$, the space $\M_{\alpha}$ contains empirical measures of matrices of size $n$ only if $\alpha n$ is an integer value. Therefore, this space is slightly too small in order to understand the convergence properties of the spectrum of sequences of random matrices asymptotically charging non-trivial mass on the real axis: the spectral density of matrices of size $n$ can only charge a mass proportional to $1/n$ to the real axis. In order to take into account these fluctuations of the mass on the real axis as $n$ is increased, we introduce the decreasing sequence of spaces:
\[\M_{\alpha}^{n}=\{\mu \in \M_1^S(\C):\vert \mu(\R)-\alpha \vert \leq \frac 1 {n} \}.\]
The limit of this sequence is exactly $\M_{\alpha}$, and for any $n\in \mathbb{N}$, the space $\M_{\alpha}^n$ contains all empirical spectral densities of matrices with a number of real eigenvalues $k=\overline{\alpha n}$. 
Throughout the paper we will use the classical L\'evy topology, which provides a metric for the topology of weak convergence\footnote{For the sake of completeness, we recall that the distance between two probability measures $(\mu,\nu)$ on $\mathbbm{C}$ is defined as:
\[d_L(\mu,\nu)=\inf\{\delta>0\,;\,\mu(A)\leq\nu(A^\delta)\mbox{ and }\nu(A)\leq\mu(A^\delta)\quad \forall A\in\mathcal{B}(\mathbb{C}) \}\]
where $\mathcal{B}(\mathbb{C}$ is the Borel algebra of $\mathbbm{C}$ and for $A\in \mathcal{B}(\mathbbm{C})$, $A^\delta=\{x;d(x,A)\leq\delta\}$.}.
We will be interested in sequences of random matrices $(M^n)_{n\geq 0}$ whose spectral density belongs to 
$\M_{\alpha}^{n}$. These matrices satisfy the following estimates:

\begin{thm}\label{thm:LDP}
	For any $\nu \in \M_\alpha$, we have:
\begin{equation*}
	\begin{cases}
		\displaystyle{\lim_{\delta \searrow 0} \; \limsup_{n\to\infty}\;\frac 1 {n^2}\log\P[\hat{\mu}^n \in \M_{\alpha}^{n}\cap B(\nu,\delta)] \leq -\mathcal{I} [\nu]}\\
		\ \\
		\displaystyle{\lim_{\delta \searrow 0} \; \liminf_{n\to\infty} \; \frac 1 {n^2}\log\P[\hat{\mu}^n \in \M_{\alpha}^{n}\cap B(\nu,\delta)] \geq -\mathcal{I}[\nu]}\ .
	\end{cases}
\end{equation*}
where $B(\nu,\delta)$ is the L\'evy ball of radius $\delta$ centered at $\nu$ and $K=\frac{3}{8}$.
\end{thm}
This property is more general than a Large Deviation Principle (LDP) in that in allows to take into account the constraint on the asymptotic proportion of real eigenvalues, which vary with the sequence index. We denote the half complex plane as $\mathbb{H}=\{z\in\C; \Im(z)>0\}$.

\begin{proof}
	The proof is based on a combination of evaluations and methods proposed in~\cite{Benarous97,Benarous98} in order to prove large deviations principles for the Wishart or Ginibre ensembles. The first inequality (upper bound) is proved directly by considering the joint eigenvalue density. \comment{For the sake of completeness, we provide a sketch of the proof in Appendix~\ref{sec:upperbound} highlighting the main differences with the one in~\cite{Benarous98} .} 

	The lower bound is slightly more complex. In~\cite{Benarous97}, the authors propose an original construction of a particular set of points on the real line, from which they construct a measure whose probability compares to the rate function and lower bounds the probability we aim at controlling. This construction was generalized in~\cite{Benarous98} where the points now belong to $\H$. For our purposes, a mixed construction both on the real line and on $\H$ proves necessary to ensure the property stated in the theorem. \comment{To this purpose, we define $l=(n-k)/2$ the number of complex eigenvalues in $\mathbb{H}$ and the measures $\nu_R(z)=\nu(z)\1_{z\in\R}$ and $\nu_C(z)=\nu(z)\1_{z\in\C\setminus\R}$.} The construction, schematically described in Fig~\ref{fig:construction}, proceeds as follows. By the continuity properties of the rate function proved in~\cite{Benarous98}, we can assume that $\nu_C$ has no atom with continuous and everywhere positive density with respect to Lebesgue's measure on $\H$, and 
we can define a square $\mathcal{H}$ in $\H$ containing at least $(1-\alpha)(1-1/n)$ of the mass.  Similarly for $\nu_R$ on $\R$, one can define a bounded interval $\mathcal{R}$ 
containing 
a mass at least equal to $\alpha(1-1/n)$. The square $\mathcal{H}$ (resp. the interval $\mathcal{R}$) can be decomposed into $D$ disjoint squares $\{B_l\}_{l\in \{1,\cdots, D\}}$ (resp. $\delta$ disjoint intervals) of length proportional to $1/\sqrt{n}$ (resp. $1/n$), in each of which are set a fixed number of complex (resp. real) points based on the density $\nu$ in each of these intervals. These points are denoted $(\tilde{\lambda}_1<\cdots<\tilde{\lambda}_k) \in \R^k$ and $(\tilde{Z}_1,\cdots, \tilde{Z}_l) \in \H^l$ (see~\cite{Benarous97,Benarous98} for the details of the construction). The important information on these points is that:
	\renewcommand{\theenumi}{(\roman{enumi})}
	\begin{enumerate}
		\item $\tilde{\lambda}_i$ are the boundaries of the intervals (\comment{remark that the $\lambda_i$ could alternatively be fixed as quantiles of the distribution $\nu_R$}),
		\item the number of points $\tilde{Z}_j$ is related to the mass contained in the square they are contained in,
		\item the distance to the boundary of the square, as well as distances between two points, are lower bounded by $C/\sqrt{n}$ for some constant $C$.
	\end{enumerate}
	\begin{figure}[htbp]
		\centering
			\includegraphics[width=.7\textwidth]{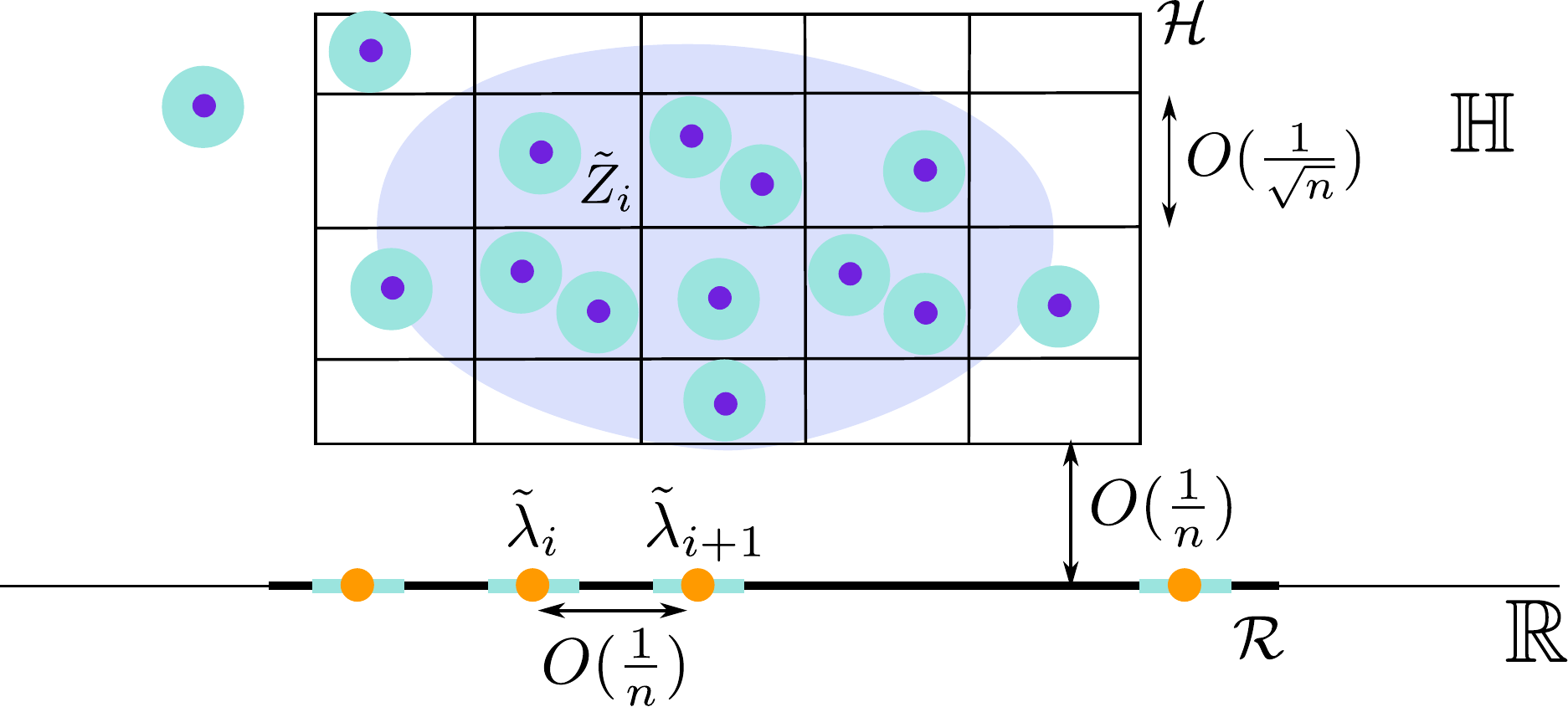}
		\caption{The artificial empirical measure constructed for the 
lower bound in the ball centered at $\nu$. The density of $\nu$ charges both $\H$ 
(light blue region) and a part on $\R$ (bold black line). The square $\mathcal{H}
\subset \H$ contains in its interior a mass greater $(1-\alpha)(1-1/n)$ and is partitioned
into smaller squares of typical size of order $1/\sqrt{n}$. Atoms $\tilde{Z}_i$ are blue circles, 
surrounded by green balls of size $\varepsilon/n$, and together form the space 
$D_{\varepsilon}$. Along the real line, orange circles are the atoms $\tilde{\lambda}_i$ and 
green intervals are of amplitude $\delta$.}
		\label{fig:construction}
	\end{figure}
	Note that such a construction yields an empirical measure whose mass is smaller that one. This is completed by adding some points outside $\mathcal{H}$, satisfying the same conditions of distance to boundaries and between points. Thus the constructed empirical distribution is close from $\nu$, in the sense that the distance (in total variation) between these two measures is arbitrarily small as soon as taking into account a sufficient number of points $(\tilde{\lambda}_i)$ and a sufficiently fine partition of $\H$. Around these points, we can define small intervals $[\tilde{\lambda}_i-\delta/2,\tilde{\lambda}_i+\delta/2]$ and for sufficiently small $\varepsilon>0$, non-overlapping balls centered at $\tilde{Z}^i$ with radius $\varepsilon/n$. The union of these balls is denoted $D^{\varepsilon}$. From this construction, following exactly the same algebra as in the pure real or complex case, we obtain: \\
	\begin{align*}
			Q^n(B(\nu,\delta)\cap \M_{\alpha}^n) \geq B_n D_n \exp\Bigg((-l +\frac 1 {\sqrt{l}}) \sum_{j=1}^{l/2} \vert \tilde{Z}_j\vert^2 +\sum_{i\neq j} \log\vert \tilde{Z}_i - \tilde{Z}_j \vert \vert \tilde{Z}_i^* - \tilde{Z}_j^* \vert - l(l+k)\log(1-\frac{2\varepsilon}{C})\\
			 - k\sum_{i=1}^k (\vert \tilde{\lambda}_i\vert +\delta)^2+\sum_{i\neq j} \log\vert \tilde{\lambda}_i-\tilde{\lambda}_j\vert + 2 \sum_{i=1}^k \sum_{j=1}^{l/2} \log\vert \tilde{\lambda}_i-\tilde{Z}_j\vert
			\Bigg)\\
			 \geq B_n D_n \exp\Bigg((k+1)^2 \int \log\vert x-y\vert d\nu_R(x)d\nu_R(y) -(k+1)^2 \int x^2 d\nu_R(x)\\
			 - l \sum_{i=1}^{l/2} \vert \tilde{Z}_j \vert^2 +\frac{l^2}{2}\int\int \log\vert x-y\vert d\nu_C(x)d\nu_C(y)+2kl\int\int \log\vert x-y\vert d\nu_C(x)d\nu_R(y) + R(\delta,\varepsilon,n)\Bigg)
	\end{align*}
	where $C$ denotes a constant independent of $n$ (possibly depending on the parameters $\delta$ and $\varepsilon$), $B_n$ a constant depending on all parameters such that $ \log(B_n)=o(n^2)$, $D_n$ a constant satisfying the same property as in the proof of the upper bound ($1/n^2 \log(D_n)\to K$), and $R(\delta,\varepsilon,n)$ a term tending to zero when $\delta$ or $\varepsilon$ tend to zero, and is negligible compared to $n^2$ as $n\to\infty$ (this term is obtained explicitly). From this expression, it is now easy to see that :
	\begin{align*}
		\frac 1 {n^2} \log(Q^n(B(\nu,\delta)\cap \M_{\alpha}^n)) & \geq \frac {\log(C_n)} {n^2} + \frac{(k+1)^2}{n^2} \Bigg(\int \vert x\vert^2 d\nu_R(x) + \int\int\log\vert x-y\vert d\nu_R(x)d\nu_R(y)\Bigg)\\
		& \qquad + \frac {l^2}{n^2}\Bigg(\int \vert x\vert^2 d\nu_C(x) + \int\int\log\vert x-y\vert d\nu_C(x)d\nu_C(y)\Bigg)\\
		& \qquad +\frac {2kl}{n^2}\int\int\log\vert x-y\vert d\nu_C(x)d\nu_C(y) + o(1)
	\end{align*}
	and therefore in the limit $n\to \infty$ we conclude that:
	\begin{align*}
		\lim_{\delta \searrow 0}\liminf_{n\to \infty} \frac 1 {n^2} \log(Q^n(B(\nu,\delta)\cap \M_{\alpha}^n))  \geq & K - \alpha \int \vert x\vert^2 d\nu_R(x) + \alpha^2 \int\int\log\vert x-y\vert d\nu_R(x)d\nu_R(y)\\
		&  - (1-\alpha)\int \vert x\vert^2 d\nu_C(x) + (1-\alpha)^2\int\int\log\vert x-y\vert d\nu_C(x)d\nu_C(y)\Bigg)\\
		&  + 2\alpha (1-\alpha)\int\int\log\vert x-y\vert d\nu_C(x)d\nu_C(y)\\
		 = &\I[\alpha \nu_R+(1-\alpha)\nu_C]
	\end{align*}
	which ends the proof.
\end{proof}

Similarly to a LDP, the above theorem readily ensures the following
\begin{thm}
The rate function $\I$ has a unique minimizer $\mu_\alpha$ on the space $\M_{\alpha}$. The sequence of empirical measures \comment{ $\hat{\mu}^n\in\mathcal{M}_{\alpha}^n$} converges in the L\'evy topology towards $\mu_\alpha$.
\end{thm}

\begin{proof}
We start by noting that the space $\M_{\alpha}$ is convex since any convex combination of elements of $\M_{\alpha}$ remains in $\M_{\alpha}$. Moreover, it is a closed subset of $\M_1^S(\C)$. \comment{The convexity of $\I$ \cite{Benarous97,Benarous98} ensures that} the map $\I$ restricted to $\M_{\alpha}$ remains lower-semicontinuous, and has compact level sets in $\M_{\alpha}$. In particular, it implies that the map $\I$ is strictly convex on $\M_{\alpha}$, guaranteeing that there exists a unique minimum $\mu_{\alpha}$ in $\M_{\alpha}$. In order to ensure that sequences of empirical processes in $\M_{\alpha}^n$ converge towards $\mu_{\alpha}$ exponentially fast 
with speed $n^2$, we further need to ensure that these sequences are tight. That property is, again, a consequence of the analogous properties proved in the case of one-dimensional log-gases~\cite{Benarous97} and on the the unconstrained Ginibre ensemble~\cite{Benarous98} (seen as a consequence of the first result). Here, the elegant proof proposed in~\cite{Benarous98} readily extends to our constrained case and allows ensuring the exponential tightness of sequences of random matrices with spectral distribution constrained to belong to $\M_{\alpha}^n$, as we now outline. Defining for $K\in [0,1]$ and $r>0$ the compact set $A_{K,r}^n$, subset of $\M_{\alpha}^{n}$, as:
\[A_{K,r}^n = \{\mu\in \M_{\alpha}^n \;;\; \mu(B_r^c) > K\},\]
where $B_r$ is the complex centered ball with radius $r$, and $A_{K,r} = \{\mu\in \M_{1}^S(\C) \;;\; \mu(B_r^c) > K\}$ it is easy to see, using the inequality $\log(\vert x-y\vert)\leq 2+(\vert x\vert^2+\vert y \vert^2)/4$, that:
\begin{align*}
	Q^n(\hat{\mu}^n\in A_{K,r}^n) &\leq g_n \int \prod{i=1}^k d\lambda_i \prod_{j=1}^l dx_jdy_j \mathbbm{1}_{A_{K,r}} \exp\left(-\frac{n^2}{4} \int (\vert z\vert^2+\vert z'\vert ^2)d\tilde{\mu}(z,z')\right)\\
	&\leq g_n \exp\left(-\frac{Kn^2 r^2} 4\right)
\end{align*} 
where $\tilde{\mu}$ is the empirical measure with atoms at $\lambda_i$ and $x_j \pm \mathbf{i} y_j$ and $\log(g_n)=O(n^2)$. This inequality readily yields the desired of exponential tightness including the constraint that empirical distributions belong to $\M_{\alpha}^n$.

We have therefore built up all necessary ingredients for proving convergence of our constrained ensembles. Indeed, classical Large Deviations theory show exponentially fast convergence of a sequence of empirical measures under the condition that the sequence is exponentially tight and satisfies a LDP with good rate function (the sequence converges towards the unique minimum of the rate function). Here, analogous properties were proved under the constraint that each element of the sequence belongs to $\M_{\alpha}^n$, and the classical proof ensures convergence of the constrained measures towards the minimum of the rate function on $\M_{\alpha}$. 

\end{proof}

This method of proof of the convergence of the sequence of empirical measures $\hat{\mu}^n$ to $\mu_{\alpha}$ the unique minimizer of $\mathcal{I}$ on $\M_{\alpha}$ goes beyond the case of constraining the Ginibre ensemble to the number of real eigenvalues, and provides an account for so-called \emph{log-gas method} which proved very efficient for the understanding of rare events in the Gaussian and Wishart $\beta$-ensembles~\cite{majumdar2009index,majumdar2012number} or the Ginibre ensemble~\cite{allez2014index}.

In our purpose to characterize Ginibre matrices with prescribed proportion of real eigenvalues, we therefore need to find the minimizer $\mu_{\alpha}$ of the rate function in $\M_{\alpha}$. If this is possible, we can access (i) to the typical distribution of eigenvalues under our constraint and (ii) to the probability of these events at leading order, logarithmically equivalent to $\exp(-n^2 \I[\mu_{\alpha}])$. Our efforts will therefore now be devoted to the characterization of these distributions.

\section{Asymptotic behavior of $p^n_{\alpha n}$}\label{sec:asymptotic_pnk}
We now use the LDP in order to characterize the asymptotic behavior of the probability $p^n_{\alpha n}$ of having asymptotically $k\sim \alpha n$ real eigenvalues. For $M^n$ a random matrix of size $n\times n$ from the real Ginibre ensemble, we recall that
\begin{equation}
	p^n_k = \P[ M^n \mbox{ has }k\mbox{ real eigenvalues}] = \P[\hat{\mu}^n(\R)=k/n]
\end{equation}
where $\hat{\mu}^n$ is the empirical spectral distribution of $M^n$. To account for the discrete nature of the mass on the real line and for the parity of the number of complex eigenvalues, we consider the following quantity:
\begin{equation}
	p^n_{\overline{\alpha n}} = \P[|\hat{\mu}^n(\R)-\alpha|\leq 1/n] = \P[\hat{\mu}^n \in \mathcal{M}_{\alpha}^n]\ .
\end{equation}
In order to characterize this quantity, it is tempting to apply directly the large deviation principle of~\cite[Theorem 1.1]{Benarous98} on the set $A=\mathcal{M}_{\alpha}$ which is closed with empty interior. One therefore obtains the upper bound:
\begin{equation}
	\limsup_{n\to \infty}\frac{1}{n^2} \log \P[\hat{\mu}^n(\R) \in A] \leq - \inf_{\nu \in A} \I[\nu] = -\I[\mu_{\alpha}]\ .
\end{equation}
Here one faces two difficulties. First, because of finite size effects the quantity $\P[\hat{\mu}^n(\R) \in A]$ is not exactly equal to $\P[\hat{\mu}^n(\R) \in \mathcal{M}_{\alpha}^n]$. In fact, if $\alpha$ is irrational $\log\P[\hat{\mu}^n(\R) \in A] = -\infty$ for any $n$. Second, even if this problem is solved, only an upper bound holds.

A much better estimate can be achieved using the results of theorem~\ref{thm:LDP}:
\begin{corollary}
For any $\alpha \in (0,1)$, we have:
\begin{equation}\label{eq:pnan}
	\lim_{n\to \infty} \frac{1}{n^2} \log p^n_{\overline{\alpha n}}=   \lim_{n\to \infty} \frac{1}{n^2} \log \P[\hat{\mu}^n \in \mathcal{M}_{\alpha}^n] =  - \mathcal{I}[\mu_{\alpha}]
\end{equation}
\end{corollary}

\begin{proof}
	 In the proof of theorem~\ref{thm:LDP} we have shown that for any $\nu \in \mathcal{M}_{\alpha}$ and $\delta>0$:
	\begin{equation}
		\limsup_{n\to \infty} \frac{1}{n^2} \log \P[\hat{\mu}^n \in \mathcal{M}_{\alpha}^n \cap B(\nu,\delta)] \leq -\inf_{\mu \in \mathcal{M}_{\alpha} \cap B(\nu,\delta)}\mathcal{I}[\mu]
	\end{equation}
	and 
	\begin{equation}
			\liminf_{n\to \infty} \frac{1}{n^2} \log \P[\hat{\mu}^n \in \mathcal{M}_{\alpha}^n \cap B(\nu,\delta)] \geq -\mathcal{I}[\nu]
	\end{equation}
	And therefore, letting $\delta \to \infty$, we obtain:
	\begin{equation}
		\limsup_{n\to \infty} \frac{1}{n^2} \log \P[\hat{\mu}^n \in \mathcal{M}_{\alpha}^n] \leq -\inf_{\mu \in \mathcal{M}_{\alpha} }\mathcal{I}[\mu]
	\end{equation}
	\begin{eqnarray}
			\liminf_{n\to \infty} \frac{1}{n^2} \log \P[\hat{\mu}^n \in \mathcal{M}_{\alpha}^n] &\geq& -\mathcal{I}[\nu] \\&\geq& -\inf_{\mu \in \mathcal{M}_{\alpha} }\mathcal{I}[\mu]
	\end{eqnarray}
	readily proving result announced in equation~\eqref{eq:pnan}.
\end{proof}

This result shows that the probabilities $p^n_{\overline{\alpha n}}$ decrease as $e^{-n^2}$. \comment{For $\alpha=1$, we know that $\mu_1$ is the semi-circular law, and $\I[\mu_1]$ can be readily computed and is equal to $\log(2)/4$ (it is exactly the difference between the constant term in the rate function for Hermitian matrices and that for non-Hermitian matrices, see e.g.~\cite{Benarous97,Benarous98}). This shows that $\log p_{\varphi(n)}^n\sim -\frac{n^2}{4}\log 2$ for any map $\varphi:\mathbbm{N}\mapsto \mathbbm{N}$ with $\varphi(n)/n\to 1$. In particular, $\log p_{n-2r}^n\sim -\frac{n^2}{4}\log 2$ for any $r\in \mathbbm{N}$}. Furthermore the continuity and convexity properties of $\alpha\mapsto\I[\mu_\alpha]$ proven below imply the continuity and convexity of the map $\alpha\mapsto \comment{\lim_{n\to\infty}} p_{\overline{\alpha n}}^n$.

\section{Properties of the minimizer $\mu_{\alpha}$}\label{sec:variational}
We investigate the qualitative properties of the distribution $\mu_{\alpha}$, the large $n$ limit spectral distribution of the Ginibre ensemble with asymptotic proportion $\alpha$ of real eigenvalues. 

\subsection{Continuity and monotonicity of $\alpha \mapsto \mu_{\alpha}$}

In this subsection, we establish continuity of the minimizer and of the minimum of the rate function on the subspaces $\M_\alpha$ with respect to the parameter $\alpha$, as well as the convexity and monotonicity of the minimum.
\begin{proposition}$ $
	\begin{enumerate}
		\item The mappings $\alpha \in [0,1] \mapsto \mu_{\alpha}$ and $\alpha \in [0,1] \mapsto \I[\mu_{\alpha}]$ are continuous.
		\item The mapping $\alpha \in[0,1] \mapsto \I[\mu_{\alpha}]$ is convex and increasing.
	\end{enumerate}	
\end{proposition}

\begin{proof}$ $
	\begin{enumerate}
		\item 	We define:
\begin{equation}\label{eq:Jalpha}
\mathcal{J}_{\alpha}[\mu]:=
\begin{cases}
	\mathcal{I}[\mu]\mbox{ if }\mu \in \mathcal{M}_{\alpha}\\
	+\infty\mbox{ otherwise.}
\end{cases}
\end{equation}

$\I$ has a unique minimizer in $\M_\alpha$. In order to show the continuity properties in $\alpha$ we will use the theory of $\Gamma$-convergence \comment{\cite[Theorem 1.21 p.29]{braides}}. In detail we consider a sequence $(\alpha_n)_{n\geq 0}\in[0,1]^\mathbb{N}$ converging towards $\alpha\in[0,1]$. For each $\alpha_n$, there exists a minimum $\mu_{\alpha_n}$ of the map $\mathcal{J}_{\alpha_n}$. To show that $\mu_{\alpha_n}\to\mu_\alpha$, we will show that the sequence of functionals $\mathcal{J}_{\alpha_n}$ $\Gamma$-converges to $\mathcal{J}_\alpha$. This amounts to proving a few inequalities on the limits of the sequence of functions $\mathcal{J}_{\alpha_n}$ together with an equi-coerceness property of the sequence $(\mathcal{J}_{\alpha_n})_n$. In detail, it suffices to prove that the sequence of processes satisfies the following properties:
\begin{enumerate}
	\item for any $\mu \in \mathcal{M}$ and any converging sequence $\alpha_n \to \alpha$ with $\mu_{\alpha_n} \to \mu$, the following inequality holds:
			\begin{equation}
				\label{eq:liminfgamma}
				\liminf_{n\to \infty} \mathcal{J}_{\alpha_n}[\mu_{\alpha_n}] \geq  \mathcal{J}_{\alpha}[\mu]
			\end{equation}
	\item  for any $\mu \in \mathcal{M}$, \comment{there exists a sequence $\alpha_n\to \alpha$ and $\nu_{\alpha_n} \in \M_{\alpha_n}$ converging to $\mu$ such that:
			\begin{equation}
				\label{eq:limsupgamma}
				\limsup_{n\to \infty} \mathcal{J}_{\alpha_n}[\nu_{\alpha_n}] \leq  \mathcal{J}_{\alpha}[\mu]
			\end{equation}}
	\item (Equicoerciveness): for all $n\in \mathbb{N}$ and for all $t>0$ there exists a compact set \comment{$K_t\subset \M^1(\mathbb{C})$} such that \comment{$\{\nu:\mathcal{J}_{\alpha_n}[\nu]\leq t\} \subset K_t$}. 
	
\end{enumerate}

We now prove that the three assumptions indeed hold in our case. 

(a) This property is a consequence of the regularity properties of good rate functions. Indeed, let $\mu \in \mathcal{M}$ be a fixed measure. If $\mu \notin \mathcal{M}_{\alpha}$, then for $n$ large enough, $\mu_{\alpha_n}$ cannot be in $\mathcal{M}_{\alpha_n}$, the above $\liminf$ is infinite and inequality~\eqref{eq:liminfgamma} is trivial. Otherwise, if $\mu\in \mathcal{M}_{\alpha}$, the above inequality is actually a direct consequence of the lower-semicontinuity of $\mathcal{I}$ : we can assume that $\mu_{\alpha_n} \in \mathcal{M}_{\alpha_n}$ (or at least a subsequence), so that $\mathcal{J}_{\alpha_n}[\mu_{\alpha_n}] = \mathcal{I}[\mu_{\alpha_n}]$ and then one uses the lower-semicontinuity of $\mathcal{I}$ to obtain \eqref{eq:liminfgamma}.
		
(b) Given a measure $\mu \in \mathcal{M}$, the second step amounts to finding a sequence $\alpha_n\to \alpha$ \comment{and a sequence of measures $\nu_{\alpha_n} \in \M_{\alpha_n}$ with $\nu_{\alpha_n}\to \mu$} such that the inequality~\eqref{eq:limsupgamma} holds. For $\mu \notin \mathcal{M}_{\alpha}$, the inequality is trivial, hence valid for any sequence. Otherwise, $\mu$ can be written as $\mu =\alpha \mu_R + (1-\alpha)\mu_C$, and for any sequence $\alpha_n \to \alpha$, one defines $\comment{\nu_{\alpha_n}} = \alpha_n \mu_R + (1-\alpha_n)\mu_C$. Then, $\comment{\nu_{\alpha_n}} \in \mathcal{M}_{\alpha_n}$ and $\mathcal{J}_{\alpha_n}[\comment{\nu_{\alpha_n}}] = \mathcal{I}[\comment{\nu_{\alpha_n}}]$. Therefore, one obtains \eqref{eq:limsupgamma} by continuity of $a \mapsto \mathcal{I}[a \mu_R + (1-a)\mu_C]$. 
		
(c) The equicoerciveness property is a consequence of the fact that $\mathcal{I}$ is a good rate function, hence has compact level level-sets. Therefore, for any $t\in \mathbb{R}$, there exists a compact $K_t$ such that $\{\mathcal{I} \leq t\} \subset K_t$. We thus have $\{\mathcal{J}_{\alpha_n}\leq t\} = \{\mathcal{I}\leq t\} \cap \mathcal{M}_{\alpha_n} \subset K_t$.
	
These three results proved the $\Gamma$-convergence of $\mathcal{J}_{\alpha_n}$ towards $\mathcal{J}_\alpha$, and therefore of $\mu_{\alpha_n}$ to $\mu_{\alpha}$ for any sequence $\alpha_n\to \alpha$, which concludes the proof of the first point of the proposition.\\

\item We now prove the convexity and monotony of $\alpha \mapsto \I[\mu_{\alpha}]$. Let $\alpha,\alpha' \in [0,1]$. Since $\mu \mapsto \I[\mu]$ is convex, we have:
\begin{equation}
	\I[t \mu + (1-t) \nu] \leq t \I[\mu] + (1-t) \I[\nu] \mbox{ for }t\in [0,1],\  \mu \in \M_{\alpha} \mbox{ and }\nu \in \M_{\alpha'}
\end{equation}
Therefore, taking the infimum on the above equation shows that $\alpha \mapsto \I[\mu_{\alpha}]$ is convex. Since, $\I[\mu_{0}]<\I[\mu_{1}]$ and $\I[\mu_{\alpha}] \geq \I[\mu_{0}]$, we conclude that the continuous mapping $\alpha \mapsto \I[\mu_{\alpha}]$ is necessarily increasing.
\end{enumerate}
	
\end{proof}

\begin{remark}
	Since $\alpha \mapsto \I[\mu_{\alpha}]$ is increasing, the minimum of the rate function conditioned on $\M_{\alpha}$ is equal to the minimum on $\M^+_{\alpha} =  \{\mu \in \M_1^S(\C): \mu(\R) \geq \alpha \}$:
	\begin{equation}
		\inf_{\M_\alpha} \I[\mu] = \inf_{\M^+_{\alpha}} \I[\mu]
	\end{equation}
	
\end{remark}

With this continuity properties in hand, we return to the characterization of the qualitative features of $\mu_\alpha$. 

\subsection{Qualitative description of $\mu_{\alpha}$}
We are now interested in characterizing the distribution $\mu_{\alpha}$. The large deviation principle defines this distribution as the minimum of $\mathcal{I}$, which can therefore be characterized through the computation of the differential of $\I$ in the space $\M_{\alpha}$. This differential is a linear operator, acting on the space of signed measures $h$ symmetrical with respect to the real axis and such that $h(\C)=0$, with traces $h^R$ on the real axis and $h^C$ on $\H$ satisfying $h^R(\R)=h^C(\H)=0$. Denoting $\alpha \mu^R_{\alpha}$ the trace of $\mu_{\alpha}$ on the real axis and $\frac{1-\alpha}{2} \mu^C_{\alpha}$ on $\H$ (with these definitions, \comment{$\mu^R_{\alpha}(\R)=\mu^C_{\alpha}(\H)=1$}), we find that:
\begin{align*}
	d_{\mu}\I[h] = \frac 1 2 \int_{\C} \vert z \vert^2 dh(z)-\int \log\vert z-z'\vert d\mu(z)dh(z')
\end{align*}
which can be rewritten as:
\begin{align*}
 \partial_{\mu^R} \mathcal{I}[\mu]=&\int\left\{\alpha \vert x\vert^2+\alpha^2\int\log\vert x-x' \vert^2 d\mu^R(x') +\alpha(1-\alpha)\int\log|x-z|^2d\mu^C(z)\right\}h^R(dx)\\
 \partial_{\mu^C} \mathcal{I}[\mu]=&\int\left\{ (1-\alpha) \vert z \vert^2 +\alpha(1-\alpha)\int\log|z-x|^2d\mu^R(x)  \right.\\
		&\left.+(1-\alpha)^2\int\left( \log|z-z'|^2+ \log|z-z'^*|^2 \right)d\mu^C(z')\right\}h^C(dz)\ .
\end{align*}
Necessarily, provided that these differentials are bounded at $\mu_{\alpha}$, the linear behavior described by the above equation fails (indeed, the freedom to replace $h$ by $-h$ would contradict the minimality of $\I$ on $\M_{\alpha}$ at $\mu_{\alpha}$). Therefore the differential operator is equal to zero for any acceptable measure $h$, implying the following proposition:

\begin{proposition}
The measure $\mu_{\alpha}$ satisfies the system of integral equations defined for any $x\in\R$ and $z\in \H$ as
\begin{align*}
&\alpha x^2+\alpha^2\int\log\vert x-x' \vert^2 d\mu^R_{\alpha}(x') +\alpha(1-\alpha)\int\log|x-z'|^2d\mu^C_{\alpha}(z')=0\\
&(1-\alpha) \vert z \vert^2 +(1-\alpha)^2\int\left( \log|z-z'|^2+ \log|z-z'^*|^2 \right)d\mu^C_{\alpha}(z')+\alpha(1-\alpha)\int\log|z-x'|^2d\mu^R_{\alpha}(x') =0.
\end{align*}
For the two extreme cases $\alpha\in\{0,1\}$ we recover the circular law $\mu^R(x)=0$, $\mu^C(z)=\frac{1}{\pi}\1_{|z|<1}$ for $\alpha=0$ and the semi-circular law $\mu^R(x)=\mu_{sc}(x)=\frac{1}{\pi}\sqrt{2-x^2}\1_{|x|\leq\sqrt{2}}$, $\mu^C(z)=0$  for $\alpha=1$.
\end{proposition}


These equations are generally too complex to be solved explicitly.  It is however possible to characterize further the minimizer $\mu_{\alpha}$ for intermediate values of $\alpha$.

 \begin{thm}\label{thm:mualpha}
 Denoting $\alpha \mu^R_{\alpha}$ and $\frac{1-\alpha}{2} \mu^C_{\alpha}$ the trace on $\R$ and $\H$ of the minimizer $\mu_{\alpha}$ of $\I$ on $\M_{\alpha}$ for $\alpha\in [0,1]$. The measure $\mu_{\alpha}$ has the following profile:
\renewcommand{\theenumi}{\roman{enumi}}
	\begin{enumerate}
		\item  $\mu^R_{\alpha}$ has a density $g_\alpha$ with respect to Lebesgue's measure on $\R$ such that $g_\alpha\in L^2(\R)$. This density vanishes outside an interval $[-R_\alpha,R_\alpha]$ for some constant $R_\alpha>0$
		\item \comment{$\mu^C_{\alpha}$ has a constant density $\frac{2}{(1-\alpha)\pi}$ with respect to Lebesgue's measure on its support $V_{\alpha}\subset \H$. This support is a simply connected bounded open set and is such that $\bar V_\alpha \cap \R = \emptyset$.}
	\end{enumerate}
\end{thm}

In \cite{armstrong2013remarks} the authors have analyzed the limiting distribution of the eigenvalues of the complex Ginibre ensemble conditioned on the event that a large proportion of the eigenvalues lies in an open subset of the complex plane. The situation we consider is similar in the sense that we have extra weight on the real line. While this is not treated explicitly in \cite{armstrong2013remarks}, the argument of proof goes through with adequate modifications.

This theorem provides us with a qualitative description of the minimizer $\mu_{\alpha}$. A more precise characterization of the minimizer requires the use of other techniques, namely numerical ones that are described below.

\section{Log-gas for Ginibre conditioned on $k$}\label{sec:particlesystem}

The present section introduces and analyses a log-gas whose equilibrium distribution is related to $\mu_\alpha$. 

The log-gas approach to spectral analysis of the real Ginibre ensemble presents two new challenges with respect to previous analysis in the literature. First, the spectra of real matrices are symmetric with respect to the real axis so one has to impose this symmetry in the gas. Second, due to the positive probability of having real eigenvalues, the log-gas is composed of two interacting phases: a 2d phase supported on the complex plane and a 1d phase constrained to the real line for all timers. This last property has a very important consequence, namely that the dynamics of the whole gas depend on the quantity of real particles $k$ with respect to the total number of particles $n$. 
Our method is based on the introduction of a mixed 1d-2d log-gas whose 
stationary distribution is that of the spectrum of real Ginibre matrices with $k$ real eigenvalues.

This system is composed of $n$ interacting particles in a two-phase gas: $k$ of these particles are confined to the real axis and the other $2l=(n-k)$ particles form $l$ pairs of complex conjugated particles outside the real axis. We denote the set of eigenvalues $\{\lambda_i\}_{1\leq i \leq n}\in\C$ with $\lambda_i=(x_i,0)$ for $1\leq i \leq k$, $\lambda_i=(x_i,y_i)$ for $k< i \leq k+l$ with $y_i>0$ and $\lambda_i=\lambda_{i-l}^*$ for $k+l<i\leq n$. 

\subsection{Log-gas equations}
To derive the equations for the gas dynamics for fixed $k$, we identify the joint probability densities \eqref{eq:pdf} as a Boltzmann factor and compute the corresponding potential, following the method of \cite{Dyson62}. One obtains the following dynamics

\begin{equation}\label{eq:2phasegas}
 d\lambda_i(t) = \frac{\sigma}{\sqrt{n}}dB_i - \sum_i\nabla_i \Phi(\lambda_1,\cdots, \lambda_n)dt
\end{equation}
where $\sigma=\sqrt{2}$, $B_i$ is a real Brownian motion with $B_i=(B^R_i,0)$ for $i\leq k$, and $B_i=(B_i^R,B^I_i)$ a complex Brownian motion for $k<i\leq k+l$. Symmetry is ensured by imposing $B_i=B_{i-l}^*$ for $k+l<i\leq n$. We denote $\nabla_i=(\partial_{x_i},\partial_{y_i})$. The potential is given by
\begin{equation}\label{eq:potential}
  \Phi(\lambda_1,\cdots, \lambda_n) =\sum_i\sum_{j\neq i}\frac1{2n} V(\lambda_i,\lambda_j) + \sum_i U(\lambda_i)
\end{equation}
with
\begin{align}
         \nonumber V(\lambda_i,\lambda_j)    &= -\log|\lambda_i-\lambda_j|\\   
                 \label{eq:U} U(\lambda_i)=& \frac{x_i^2}{2}-\1_{y_i\neq 0}\left( \frac{y_i^2}{2} + \frac{1}{2n}\log\left(\textrm{erfc}(|y_i|\sqrt{2n})\right) \right)
\end{align}
We will be specially interested on the $t\to\infty$ limit of the empirical measure
\begin{equation}
 \rho_k^n(t)=\frac{1}{n}\sum_{i=1}^n\delta(\lambda_i(t))\ .
\end{equation}

The system \eqref{eq:2phasegas} is well posed: we prove below that solutions neither collide nor blow up. The result is actually more general and extends to particle systems within more general confining potential $U$ we state:
\begin{thm} Suppose $\lambda_i(0)$ for $i=1,\dots,k+l$ are distinct and  $\lambda_i(0)=\lambda_{i-l}^*(0)$ for $k+l<i\leq n$, that $U$ is $C^2$ and the maps $(x,y)\mapsto (-x \partial_x U(x+\mathbf{i}y),-y \partial_yU(x+\mathbf{i}y))$ \comment{and $(x,y)\mapsto \Delta U(x,y)$} are upper bounded in $\R^2$ by some constant $\gamma$. The processes $\lambda_i$ are defined by \eqref{eq:2phasegas} up to the stopping time \[T=\inf\{t:\lambda_i(t)=\lambda_j(t)\mbox{ for some }j\neq i\mbox{ or }y_i=0\mbox{ for some }i>k\mbox{ or }\lambda_i=\infty\mbox{ for some }i \}\ .\] 

If $\sigma^2\leq 2$, then $\P[T=\infty]=1$: there is no collision or explosion and the particle system is defined for all times.

\end{thm}

\begin{proof}
The proof follows the one in \cite{Rogers93}.
To prove that there are no collisions we show that the drift term in \eqref{eq:2phasegas} is bounded \comment{from above}. We apply It\^o's formula on $\Phi$ to obtain
\begin{equation}
 d\Phi=\frac{\sigma}{\sqrt{n}}\sum_j \nabla_j\Phi dB_j + \sum_j \left[\frac{\sigma^2}{2n}\Delta_j\Phi  -   (\nabla_j\Phi)^2\right]dt
\end{equation}
where $\Delta_j=\partial_{x_j}^2+\partial_{y_j}^2$. Developing the drift term we have
\begin{align*}
 &\sum_j \left[\frac{\sigma^2}{2n}\Delta_j\Phi  -   (\nabla_j\Phi)^2\right]\\
 =&\sum_j \left[\frac{\sigma^2}{2n} \sum_{r\neq j} \frac1n\Delta_j V(\lambda_j,\lambda_r) + \frac{\sigma^2}{2n}\Delta U(\lambda_j)\right.\\
 &\left.- \sum_{r\neq j}\frac{1}{n^2}(\nabla_j V(\lambda_j,\lambda_r))^2 - (\nabla U(\lambda_j))^2 - \sum_{r\neq j} \frac1n\nabla_j V(\lambda_j,\lambda_r)(\nabla U(\lambda_j)-\nabla U(\lambda_r))\right]
\end{align*}
\comment{The term $- (\nabla U(\lambda_j))^2$ is clearly upperbounded. Moreover, our assumptions ensure that both terms $\frac{\sigma^2}{2n}\Delta U(\lambda_j)$ and the last term are upperbounded.} Therefore, the only term that we have to control is
\begin{equation*}
 \sum_j\sum_{r\neq j}\frac{\sigma^2}{2n^2}\Delta_j V(\lambda_j,\lambda_r)-\frac{1}{n^2}(\nabla_j V(\lambda_j,\lambda_r))^2\ .
\end{equation*}
Because of the noise in the system, the collisions between real and complex particles as well as between complex non conjugated particles are impossible. On the other hand, because of the confinement of some of the particles to the real line and the symmetry with respect to the real axis, collisions between two real particles and between a pair of complex conjugate particles are possible. Developing the Laplacian and the gradient squared for those interactions (real-real and complex-complex conjugate) one can see that
\begin{align*}
\frac{\sigma^2}{2n^2}\Delta_j V(\lambda_j,\lambda_r)-\frac{1}{n^2}(\nabla_j V(\lambda_j,\lambda_r))^2&\leq  \frac{\sigma^2-2}{2n^2}\frac{1}{|\lambda_j-\lambda_r|^2}\qquad\mbox{ if }\lambda_j,\lambda_r\in\R\\
\frac{\sigma^2}{2n^2}\Delta_j V(\lambda_j,\lambda_j^*)-\frac{1}{n^2}(\nabla_j V(\lambda_j,\lambda_j^*))^2&\leq  \frac{\sigma^2-2}{2n^2}\frac{1}{|\lambda_j-\lambda_j^*|^2}\qquad\mbox{ if }\lambda_j\in\C
\end{align*}
and hence, using the same arguments as in \cite{Rogers93}, one can show that for $\sigma^2\leq 2$ particles do not collide before an explosion.\\
Now it remains to show that there is no explosion in finite time. Let $R_t = \frac{1}{2n}\left(\sum_j x_j^2+y_j^2\right)=\avg{\mu_t^n,f}$ where $f(\lambda)=|\lambda|^2/2$. We have to prove that $R_t<\infty$ for all $t$.  
As shown in \cite{Rogers93} $R_t\leq R'_t$ a.s. for all $t$ where $R'_t$ solves
\begin{equation*}
  dR'_t = \frac{\sigma}{n}\sqrt{2R'_t}dW_t+ \left(\frac{\sigma^2}{2n} + \frac{(n-1)}{n} - \gamma\right)dt
\end{equation*}
for $W_t$ a Brownian motion. $R'_t$ is a squared Bessel process which is known not to explode so $R_t$ does not explode either.

\end{proof}

\begin{remark}
	The assumptions on the confining potential $U$ of the above theorem are clearly satisfied for $U$ given by~\eqref{eq:U}.
\end{remark}

The result of this theorem ensures existence and uniqueness of the particle system for all times and finite $n$. If this system converges in time towards a non-singular stationary distribution, the limit is necessarily given by the joint pdf~\eqref{eq:pdf} of the eigenvalues of the Ginibre ensemble under our constraints. In this case the empirical measure of the particle system $\rho^n_k(t)$ converges in distribution as $t\to\infty$ to the empirical spectral measure $\hat\mu^n$ of the Ginibre ensemble with $k$ real eigenvalues. Therefore, simulating the long term dynamics of the particle system will provide an approximation of $\hat\mu^n$.

\subsection{Numerical simulations}

In principle, it is possible to generate numerically the spectra of Ginibre matrices conditioned on $k$ by drawing Ginibre matrices at random and classifying them according to the number of real eigenvalues. However, for the events we consider, $k$ deviates from its expected value and $p^n_k$ decrease as  $e^{-n^2}$, so that matrices with large $k$ are not readily accessible with such a strategy. As we dispose of an explicit form of the joint probability distribution of the eigenvalues under our constraint, standard sampling methods, such as the Monte Carlo algorithm described in appendix \ref{sec:MC}, allow to access directly these rare events. Unfortunately, they are not efficient for large matrices. 

The same limitations at large $n$ arise for simulations of the log-gas given that the system is stochastic and one needs to average over several runs to smooth out fluctuations. To overcome these, we took advantage of the fact that the noise term perturbing each particle vanishes. Therefore, for estimating $n \to \infty$ spectral distributions, instead of stochastic simulations, we ran deterministic simulations of the gas. Even though the latter do not correspond to the eigenvalues of actual matrices, their $n\to\infty$ limit converges to the limit distribution of the eigenvalues, which is exactly what we seek to approximate. We checked convergence by increasing progressively the number of particles of the gas until the shapes for the complex support and the distribution of the real phase were visually indistinguishable.

Here we summarize our key findings, which can be considered educated guesses based upon known facts and numerical evidence (see figure \ref{fig:contours}). The first claim is that the border of the support of each connected component of the complex part of $\mu_\alpha$ is a smooth closed curve for all $0<\alpha<1$. This closed curve varies continuously with $\alpha$. In fact, as $\alpha$ goes from zero to one, the support $v$ decreases monotonically (in the sense of strict inclusion) from the unit semi-circle in $\H$ to the point $z^*=(0,y^*)$ where $z^*$ is found by minimizing the potential that a single pair of complex particles would experience interacting with $\mu^R=\mu_{sc}$ which is given by the map $z\mapsto|z^2|-\int\log|z-x|d\mu_{sc}(x)$. The coordinate $y^*$ is hence the solution of:
 	\begin{equation}\label{eq:ystar}
  			\int \frac{1}{x^2+(y^*)^2}d\mu_{sc}(x)=2
 		\end{equation}
Coincidentally, the support of $\mu^R_\alpha$ grows monotonically from $[-1,1]$ to $[-\sqrt{2},\sqrt{2}]$ as $\mu^R$ converges to the semi-circular law in $[-\sqrt{2},\sqrt{2}]$.

\begin{figure}[H]
\centering
 \includegraphics[width=0.7\textwidth]{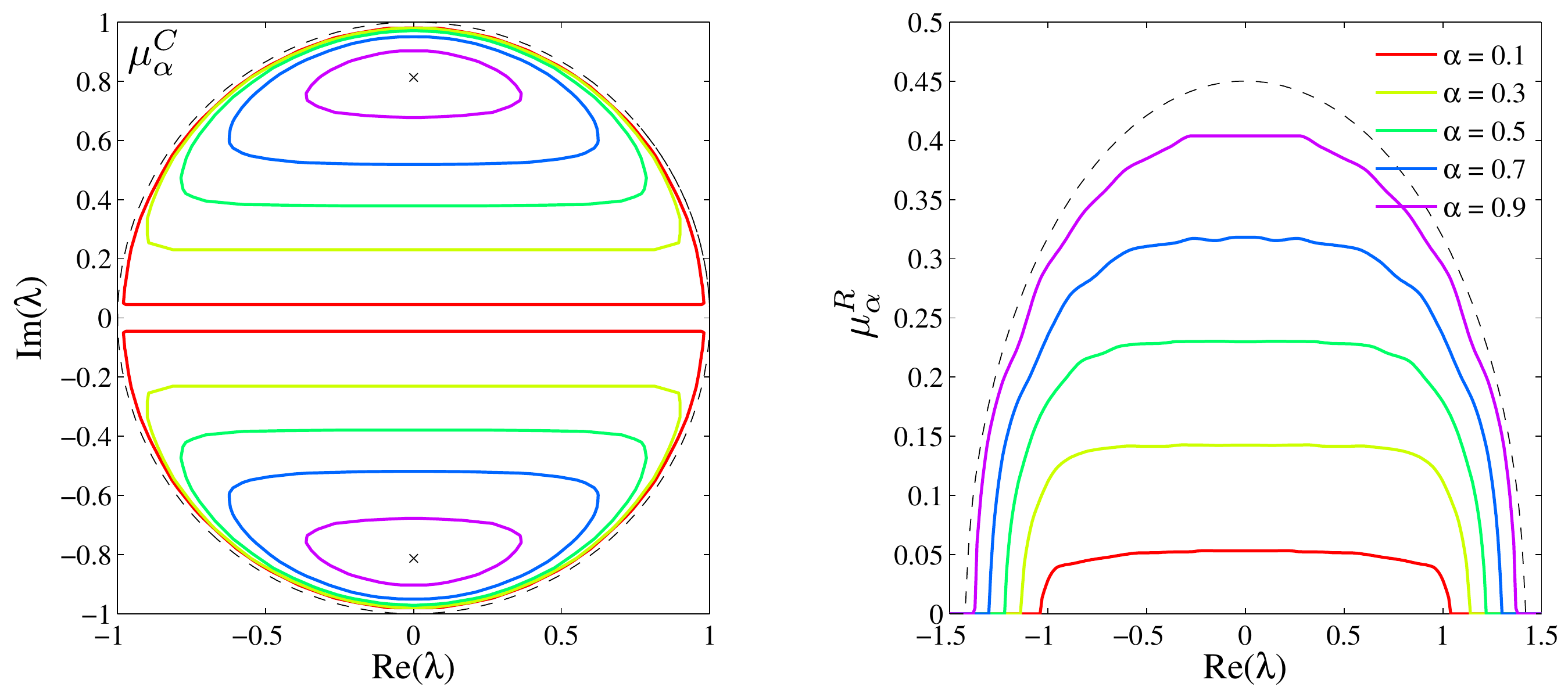}
 \caption{Left: Numerical estimation of the support of $\mu_\alpha^C$ for $\alpha\in\{0.1,0.3,0.5,0.7,0.9\}$. The crosses correspond to the points $(0,\pm y^*)$, the dashed line is the unit circle. Right: Numerical estimation of $\mu^R_\alpha$ for $\alpha\in\{0.1,0.3,0.5,0.7,0.9\}$. The dashed line is the semi-circular law. The estimations are long time simulations of log-gases with $n=1000$.}\label{fig:contours}
\end{figure}

\section{Renormalized energy and microscopic organization at zero temperature}
\label{sec:renorm}

The gas provides a precise description of the minima of the potential~\eqref{eq:potential} even at temperatures that differ from that corresponding to the eigenvalues of the Ginibre ensemble. In particular, expansion of the energy of the system supplies further information on the microscopic organization of the particles at vanishing temperature. To this end, following the work for log-gases in dimension one~\cite{sandier20131d}, two~\cite{Sandier13} or higher~\cite{rougerie2013higher}, one shall compute the next-to-leading order terms of the energy. These terms correspond essentially to the microscopic arrangements of the particles. Since we have shown that in the regime we consider with $k=\overline{\alpha n}$, there exists a  macroscopic gap between the real axis and complex eigenvalues, the energy related to the microscopic interactions of real onto the complex eigenvalues and reciprocal forces vanish in the thermodynamic limit. Once this has been noted, a direct application of the 
results in one and two dimensions~\cite{sandier20131d,Sandier13} leads to state the following:

\begin{proposition}\label{pro:renormalized}
	For large $n$, the equivalent energy of a Ginibre matrix conditioned on having $k\sim \alpha n$ real eigenvalues enjoys the following expansion around ${\mu}_{\alpha}=\alpha \mu^{R}_{\alpha} + (1-\alpha) \mu^{C}_{\alpha}$ the minimizer of the macroscopic energy:
	\begin{multline}\label{eq:ExpansionEnergy}
		\mathcal{E}_{n,\alpha}= n^2 \mathcal{I}({\mu}_{\alpha}) - \frac{(1+\alpha)}{2} n\log (n) 
		+ n\Big[ (1-\alpha)\frac{\kappa_2}{2\pi} + \alpha \frac{\kappa_1}{\pi} \Big ] \\
		- n\Big[(1-\alpha) \int_{\C} d\mu^{C}_{\alpha}\log(\mu^{C}_{\alpha})+\alpha \int_{\R} d\mu^{R}_{\alpha}\log(\mu^{R}_{\alpha})\Big] + o(n)
	\end{multline}
where the $\kappa_2$ and $\kappa_1$ are universal constants related to the dimension of the spaces where $\mu_\alpha^C$ and $\mu_\alpha^R$ are supported.
\end{proposition}
From this expansion, the minimization at zero temperature leads to state that the real eigenvalues crystallize~\cite{sandier20131d}, and to the conjecture that the complex eigenvalues organize in a regular triangular lattice (called Abrikosov lattice in the superconductivity domain~\cite{abrikosov:57}), which is proved under the assumption that the organization is a regular lattice~\cite{Sandier13}.

\section{Conclusion}
\label{sec:conc}

The main conclusions of our work can be outlined as:
\begin{enumerate}
\item Despite the fact that the joint probability distribution of real Ginibre matrices is always given the same compact formula above (Eq.\eqref{eq:pdf}), the limit distributions of the empirical measure $\hat \mu^n$ strongly depends on the number $k$ of real eigenvalues. When $k/n\rightarrow\alpha$ and $n\to \infty$, we prove that the empirical measure has a unique limit $\mu_\alpha$ that significantly departs from the circular law. 

\item The key method in establishing the above result is an LDP theorem devised to take into account both real and complex eigenvalues. Previous LDPs discarded real eigenvalues, as in the unconstrained matrices, the fraction of real eigenvalues tends to zero. While we have provided the proof in the specific case of the real Ginibre ensemble, it can readily be adapted to other situations, such as gases in higher dimensions with more general confining and repulsive potentials such as \cite{chafai13} or heterogeneous gases \cite{garcia15} with constraints (fraction of particles restrained to a subspace). It also extends to Gaussian $\beta$ and Wishart ensembles conditioned on rare events (such as anomalous proportions of eigenvalues in a given interval). In this sense, our approach  provides a theoretical basis for the log-gas method used in~\cite{Majumdar07,majumdar2009index,majumdar2012number,allez2014index}. These generalizations allow, for instance, to analyze the impact of anomalously large numbers of real 
eigenvalues on 
the spectra of random asymmetric matrices with some level correlations in the entries such as those in \cite{Sommers88,Lehmann91}.


\item One of the consequences of the LDP was to provide an asymptotic for $\log p^n_{\overline{\alpha n}}$ which scales as $-n^2\I[\mu_\alpha]$. To our knowledge, this is the first derivation of large $n$ estimate for $p^n_k$. 

\item The theoretical and numerical analysis of $\mu_\alpha$ established that, unlike the circular law which is supported by a disk, the measure $\mu_\alpha$ has two distinct components. The first is supported by a compact set of the complex plane that is well separated from the real line, and upon which $\mu_\alpha$ has uniform density of mass $1-\alpha$. The second is supported by a segment in the real plane and has a density w.r.t. Lebesgue's measure with mass $\alpha$. As $\alpha$ increases to one, the support of the former shrinks to a single point and its complex conjugate whereas the latter tends to the semi-circular law on $[-\sqrt{2},\sqrt{2}]$. 

\item The microscopic characterization of the particle distributions at zero temperature through the renormalized energy reveals that, (i) in the complex plane, particles organize in an Abrikosov lattice, similar to the unmixed 2d gas, yet (ii) on the real line, they are crystallized similarly to the Gaussian Orthogonal Ensemble in the zero temperature limit, but unlike real eigenvalues of unconstrained real Ginibre matrices.
\end{enumerate}

The last two points above establish that the real Ginibre ensemble constrained by $k/n\rightarrow \alpha$ interpolates between the circular and semi-circular law as $\alpha$ shifts from zero to one at the macroscopic level. At the microscopic level, it is a mixture of the two extreme cases. This interpolation is distinct from the ones in which the circular law is progressively flattened on the real line with intermediate elliptic like support for the spectra \cite{Sommers88,Lehmann91}. It reveals some of the rich characteristics of real random matrices due to their spectrum containing both real and complex eigenvalues.

\appendix
\section{Upperbound for the Large Deviations Principle}\label{sec:upperbound}
We now sketch the proof of the upper bound of the Large Deviations estimates:
\begin{proof}
	We define $l=(n-k)/2$ the number of complex eigenvalues in $\mathbb{H}$. Similarly to~\cite{Benarous98}, by a direct use of the joint pdf~\eqref{eq:pdf}, we find that the mass of a ball $B(\nu,\delta)\cap \M_{\alpha}^n$ is:
\comment{	
	\begin{eqnarray}
		&Q^n(B(\nu,\delta)\cap \M_{\alpha}^n)=\int_{\R^k \times \H^l}\prod_{i=1}^k d\lambda_i\prod_{j=1}^l dX_j\prod_{k=1}^l dY_k\mathbbm{1}_{B(\nu,\delta)\cap \M_{\alpha}^n} \mathbb{P}[\lambda_1,\dots,\lambda_n] \label{eq:Qn}\\		
		  \nonumber & \leq D_n \int_{\R^k \times \H^l} \prod_{i=1}^k d\lambda_i\prod_{j=1}^l dX_j\prod_{k=1}^l dY_k \mathbbm{1}_{B(\nu,\delta)\cap \M_{\alpha}^n}\exp\Big( -\frac {n^2} 2 \int_{\C^2} f_\omega(z,z') d\hat{\mu}^n(z)d\hat{\mu}^n(z') \Big)e^{Kn/2}
	\end{eqnarray}}
	where 
	\[f_\omega(x,y) = \min\left(\omega, \frac{\vert x\vert^2+\vert y\vert^2}{2} -\log\vert x-y\vert\right),\]
	the measure $\hat{\mu}^n$ defined as: 
	\[\hat{\mu}^n = \frac 1 n \left(\sum_{i=1}^k \delta_{\lambda_i} + \sum_{j=1}^{l} \delta_{X^j+\mathbf{i} Y^j}+\delta_{X^j-\mathbf{i} Y^j}\right),\]
	and the coefficient $D_n$ is such that $1/n^2 \log(D_n)\to K$. We hence have:
	\begin{multline*}
		Q^n(B(\nu,\delta)\cap \M_{\alpha}^n) \leq D_n \int_{\R^k \times \H^l} \prod_{i=1}^k d\lambda_i\prod_{j=1}^l dX_j\prod_{k=1}^l dY_k \\
		\mathbbm{1}_{B(\nu,\delta)\cap \M_{\alpha}^n}\exp\Big( -\frac {n^2} 2 \inf_{\mu\in B(\nu,\delta)\cap \M^n_{\alpha}}\int_{\C^2} f_\omega(z,z') d{\mu}(z)d{\mu}(z') \Big)e^{Kn/2}
	\end{multline*}
	allowing to conclude that:
	\[\frac 1 {n^2} \log(Q^n(B(\nu,\delta)\cap \M_{\alpha}^n)) \leq \frac{\log(D_n)}{n^2}-\frac 1 2 \inf_{\mu\in B(\nu,\delta)\cap \M^n_{\alpha}}\int_{\C^2} f_\omega(z,z') d{\mu}(z)d{\mu}(z') + o(1).\]
	The first term in the right hand side converges towards a constant. The main difference with the proof of~\cite{Benarous98} is the fact that the right hand side of the inequality depends on $n$. Here, the sequence of real values $(\inf_{\mu\in B(\nu,\delta)\cap \M^n_{\alpha}}\int_{\C^2} f_\omega(z,z') d{\mu}(z)d{\mu}(z'))_{n\geq 0}$ is non-decreasing because of the inclusion property of the spaces $\M_{\alpha}^n$. It is upper bounded by $\inf_{\mu\in B(\nu,\delta)\cap \M_{\alpha}}\int_{\C^2} f_\omega(z,z') d{\mu}(z)d{\mu}(z')$, and therefore converges towards this maximal value. This allows to conclude that:
	\[\limsup_{n\to \infty}\frac 1 {n^2} \log(Q^n(B(\nu,\delta)\cap \M_{\alpha}^n)) \leq K + \inf_{\mu\in B(\nu,\delta)\cap \M_{\alpha}}\int_{\C^2} f_\omega(z,z') d{\mu}(z)d{\mu}(z')\]
	and by continuity of the map $\mu\mapsto \int_{\C^2} f_\omega(z,z') d{\mu}(z)d{\mu}(z')$ in the weak sense, we eventually have:
	\[\lim_{\delta \searrow 0} \; \limsup_{n\to\infty}\;\frac 1 {n^2}\log\P[\hat{\mu}^n \in \M_{\alpha}^{n}\cap B(\nu,\delta)] \leq  K + \int_{\C^2} f_\omega(z,z') d{\nu}(z)d{\nu}(z').\]
	Our result follows by letting $\omega$ to infinity and using the monotone convergence theorem. 
\end{proof}

\section{Monte Carlo algorithm for the eigenvalues}\label{sec:MC}
An efficient method to approximate numerically the minimizer $\mu_{\alpha}$ and the probability distribution of the proportion of real eigenvalues is to use the Metropolis-Hastings Monte Carlo algorithm. This method consists in constructing an ergodic Markov chain whose stationary distribution is given by~\eqref{eq:pdf}. Here, we evolve a $n$-particles system $z_t$, but in contrast to the log-gas, the dynamics is now discrete, and the transition probability is based on the pdf~\eqref{eq:pdf}: a new configuration $z^*$ is drawn by modifying one of the eigenvalues at random and the Markov chain has a transition towards $z^*$ if $Q^n(z^*)>Q^n(z_t)$, and otherwise according to a Bernoulli variable of parameter $\frac{Q^n(z^*)}{Q^n(z_t)}$.

When conditioning on very rare events, (here for instance, a fixed number of real eigenvalues), cases satisfying the constraints have an extremely low probability of being explored, and more refined methods need to be developed in order to access these probabilities. In the present case, the problem is considerably simplified since we dispose of an explicit form of the distribution of the eigenvalues under our constraint. Indeed, the joint probability distribution of Ginibre matrices of size $n$ constrained on having $k$ real eigenvalues $(\lambda_i\;;\; i=1\cdots k)$ (and therefore $l=(n-k)/2$ pairs of complex eigenvalues $(z_i, i=1\cdots n-k)$) is given by:
\[\P[\lambda_1\cdots \lambda_k, z_1,\cdots,z_{n-k}]=\tilde{C}_n \prod_{i>j}\vert\lambda_i-\lambda_j\vert \prod_{i>j} \vert z_i-z_j\vert \left(\prod_{i=1}^k\exp(-\lambda_i^2)\prod_{i=1}^{n-k} \exp(-z_i^2)\textrm{erfc}(\vert z_i-z_i^*\vert/\sqrt{2})\right)^{1/2}.\] 
Classical Metropolis-Hastings algorithm with Gaussian transitions preserving the nature of the system therefore allow to access directly the distribution of eigenvalues and the probability $p(n,k)$ of the event considered.

\end{document}